\documentclass[conference]{IEEEtran}
\usepackage{theorem}
\usepackage{amsfonts}
\usepackage{amsmath}
\usepackage{amssymb}
\usepackage{graphicx}
\usepackage{epsfig}
\newtheorem{claim}{Claim}
\newcommand{\one}{{\bf 1}}
\renewcommand{\vec}{\overrightarrow}

\begin{document}
%\conferenceinfo{SAC'10}{March 22-26, 2010, Sierre, Switzerland.}
%\CopyrightYear{2010}
%\crdata{978-1-60558-638-0/10/03}

\author{Aubin Jarry, Pierre Leone and Jos\'e Rolim}
\title{VRAC: Theory \#1}
%{Virtual Raw Anchor Coordinates:\\a New Localization Paradigm}

%\numberofauthors{3}
%\author{
%\alignauthor
%Aubin Jarry\\
%       \affaddr{University of Geneva}\\
%       \affaddr{Battelle A, route de Drize 7}
%       \affaddr{CH-1227 Carouge, Switzerland}\\
%       \email{aubin.jarry@unige.ch}
% 2nd. author
%\alignauthor
%Pierre Leone\\
%       \affaddr{University of Geneva}\\
%       \affaddr{Battelle A, route de Drize 7}
%       \affaddr{CH-1227 Carouge, Switzerland}\\
%       \email{pierre.leone@unige.ch}
% 3rd. author
%\alignauthor
%José Rolim\\
%       \affaddr{University of Geneva}\\
%       \affaddr{Battelle A, route de Drize 7}
%       \affaddr{CH-1227 Carouge, Switzerland}\\
%       \email{jose.rolim@unige.ch}
%}
%\date{15 September 2009}

\maketitle

\begin{abstract}
In order to make full use of geographic routing techniques developed for sensor networks, nodes must be localized. However, traditional localization and virtual localization techniques are dependent either on expensive and sometimes unavailable hardware (e.g. GPS) or on sophisticated localization calculus (e.g. triangulation) which are both error-prone and with a costly overhead.

Instead of actually localizing nodes in the physical two-dimensional Euclidean space, we use directly the raw distance to a set of anchors to produce multi-dimensional coordinates. We prove that the image of the physical two-dimensional Euclidean space is a two-dimensional surface, and we show that it is possible to adapt geographic routing strategies on this surface, simply, efficiently and successfully.

\end{abstract}

%\keywords{Sensor networks, Routing, Localization, Virtual coordinates}

\section{Introduction}
If in wired networks each node is equipped with substantial computation and storage resources, this is not the case for sensor networks which are made of small and cheap devices and therefore can not maintain routing tables. Instead of using routing tables, local routing techniques have been developed. A compelling technique -- geographic routing -- consists in using nodes' coordinates. Many algorithms have been devised, such as GPSR \cite{KK00} and  OAFR \cite{KWZ08} which use a combination of greedy routing and perimeter routing. One can also cite GRIC \cite{PS07} a greedy routing algorithm following the sides of an obstacle when one is met, and which introduces some inertia in the direction followed by the message. Early obstacle detection algorithms, that use of a bit of memory at each node, have been proposed in \cite{HJM+09,ML+08}.

Unfortunately, obtaining physical coordinates is problematic in itself. The hypothesis of having a GPS for each sensor arguably leads to too expensive and heavy devices. This assumption may be weakened by equipping some nodes called anchors with the GPS or with other localization hardware. Approximate coordinates are then computed for all nodes of the network in a localization phase. In \cite{LR03}, three such algorithms are compared, namely Ad-hoc positioning, Robust positioning, and N-hop multilateration. One can also cite the algorithm At-Dist \cite{BKS08}, which is a distributed algorithm estimating the position of each node together with an estimate of its accuracy. Some authors improved these results by using angle measurements \cite{BGJ05}.
All of these localization techniques invariably require a flooding from the anchors and many computations at each node. They are therefore energy consuming, error-prone and compute only approximate coordinates.

Another approach followed in \cite{CC+05,MLRT04,MobileAssist_INFOCOM2005,RPSS03} consists in computing virtual coordinates and has the advantage of not needing anchors. In \cite{MobileAssist_INFOCOM2005}, the authors use a mobile unit to assist in measuring the distance between nodes in order to improve accuracy. The algorithm proposed in \cite{CC+05} first chooses three nodes that will behave as anchors and from which virtual coordinates will then be determined. If these techniques do not need any external hardware, they also suffer from their inaccuracy or high energy consumption in a preprocessing phase.

In this paper, we discard any preprocessing technique and propose to directly use raw distance information. We study routing algorithms using directly the distance to the anchors as coordinates, as first proposed in ~\cite{Vrac1}, without computing from them 2-dimensional coordinates. In section~\ref{sec_implementation} we precisely describe how the idea is implemented, and in section~\ref{sec_algebra} we analyze how a message sent towards a destination performs in the new coordinate system.

\section{Implementation}
\label{sec_implementation}

Current localization methods rely on raw information computed externally from normal sensing nodes (exact location of some anchors), and on raw information computed locally in normal sensing nodes (distance to anchors, angle measurements). In this paper, we do use the information about the distance to some anchors, but we completely discard any physical information that the anchors might have. This gives much more flexibility in the way sensor networks are deployed:
anchors might be some external entities, as a plane or a robot; anchors might be specialized nodes whose only purpose is to emit a strong signal, or they might be randomly chosen sensors which advertise their distance to the other nodes.

We build a multi-dimensional coordinate system using directly the raw information, i.e. the distance to the anchors. 
Given a node at location $X$, we define the multi-dimensional coordinates $f(X)$ of this node as its distance to the anchors at location $A_1, A_2, \dots A_n$:

$$f:X\rightarrow \begin{pmatrix} d(X,A_1) \\ d(X,A_2) \\ \dots \\ d(X,A_n) \end{pmatrix}.$$

We call this function the anchor coordinates function, and we call these multi-dimensional coordinates the anchor coordinates. Whereas any distance function, such as hop count, may be used~\cite{Vrac1}, in section~\ref{sec_algebra} we pay a special attention to the properties of $f$ when $d$ is the Euclidean distance.

In the next subsection we discuss the computation costs that are specific to using multi-dimensional coordinates. We then go into the details of greedy routing implementation, and into the details of rotating multi-dimensional vectors.

\subsection{Computation Cost}
While saving on initialization overhead, multi-dimensional routing causes some additional computation costs when sending messages compared to traditional two-dimensional routing. Here is a break-down of various vector operations:

\begin{center}
\begin{tabular}{|c|c|c|}
\hline
\bf Operation & \bf $n$-dimensional & \bf 2-dimensional\\
\hline
$\vec{u}+\vec{v}$ & $n$ additions & 2 additions\\
\hline
$k\vec{u}$ & $n$ multiplications & 2 multiplications\\
\hline
$\vec{u}\cdot\vec{v}$ & $n$ multiplications & 2 multiplications\\
& $n-1$ additions & 1 addition\\
\hline
& 1 sqrt extraction & 1 sqrt extraction \\
$\frac{\vec{u}}{||\vec{u}||}$ & 1 inversion & 1 inversion \\
& $2n$ multiplications & 2 multiplications \\
& $n-1$ additions & 1 addition\\
\hline
\end{tabular}
\end{center}

Note that additions and multiplications typically use 1 CPU cycle, whereas the expensive operations (square root extraction, inversion) stay the same in multi-dimensional routing as in traditional two-dimensional routing.
We also point out that theses computation costs are not communication costs and are lower in terms of energy consumption by some order of magnitude. 
 
\subsection{Greedy Routing}

Greedy routing is the most basic geographic routing algorithm. It consists in following the direction to the destination. This basic strategy is widely used as a default mode in most geographic routing protocols.
When a node at location $X$ which wants to send a message towards a final destination at location $D$, three implementations of greedy routing are routinely used:
\begin{enumerate}
\item {\em (canonical)} for each neighbor location $X'$, compute the distance $d(X',D)$ and send the message to the neighbor which is closest to $D$. Alternatively, compute $\vec{X'D}\cdot\vec{X'D}$ instead of $d(X',D)$.
\item for each neighbor location $X'$, compute the scalar product $\vec{XX'}\cdot\vec{XD}$ and select the neighbor with the best result.
\item for each neighbor location $X'$, compute the scalar product $\frac{\vec{XX'}}{||\vec{XX'}||}\cdot\vec{XD}$ and select the neighbor with the best result.
\end{enumerate}

These three implementations are valid for any number of coordinates.

\subsection{Rotation}

When greedy strategies fail, a number of two-dimensional routing algorithms fall back on more sophisticated routing modes that use rotations or angle computations~\cite{KK00, PS07}. When using two dimensions, a rotation is typically defined by $rot_\alpha : (x,y) \rightarrow (x\cos\alpha+y\sin\alpha, y\cos\alpha-x\sin\alpha)$. We can't define such a rotation in $n$ dimensions ($n\geq 3$). However, if we assume that our sensors were on a two-dimensional physical plane in the first place, then they are distributed over a two-dimensional surface in the multi-dimensional space (more on this in section~\ref{sec_algebra}). We do the following:
\begin{enumerate}
\item compute an orthonormal basis $(\vec{i},\vec{j})$ of the tangent plane in $f(X)$ (see section~\ref{sec_algebra}).
\item express vectors $\vec{u}$ as $x_u\vec{i} + y_u\vec{j} + \vec{\epsilon_u}$ by computing $x_u = \vec{u}\cdot\vec{i}$ and $y_u = \vec{u}\cdot\vec{j}$. We assume that $\vec{u}$ is close to the tangent plane in $f(X)$, which means that we ignore in fact $\vec{\epsilon_u}$.
\end{enumerate}
Rotations are then normally carried out on the tangent plane. The sensitive part is to compute $(\vec{i},\vec{j})$ and to make sure that the orientation of the surface is preserved when routing the message (taking the surface upside-down has the undesirable effect of negating angles). Given a node at location $X$, a destination at $D$, and a basis $(\vec{i_{old}},\vec{j_{old}})$ inherited from a previous node, we do the following:
\begin{enumerate}
\item choose two neighbors at position $X_1$ and $X_2$
  \begin{itemize}
  \item either arbitrarily (low quality, inexpensive)
  \item or such that $\frac{|\vec{XX_1}\cdot\vec{XX_2}|}{||\vec{XX_1}|| ||\vec{XX_2}||}$ is minimal (i.e. choose $\vec{XX_1}$ and $\vec{XX_2}$ as orthogonal as possible)
  \end{itemize}
\item compute $\vec{i} = \frac{\vec{XX_1}}{||\vec{XX_1}||}$
\item compute $\vec{u} = \vec{XX_2} - (\vec{i}\cdot\vec{XX_2})\vec{i}$
\item compute $\vec{v} = \frac{\vec{u}}{||\vec{u}||}$
\item compute $\sigma = (\vec{i}\cdot\vec{i_{old}})(\vec{v}\cdot\vec{j_{old}}) - (\vec{i}\cdot\vec{j_{old}})(\vec{v}\cdot\vec{i_{old}})$.
\item if $\sigma\geq 0$ then set $\vec{j} = \vec{v}$, else set $\vec{j} = -\vec{v}$.
\end{enumerate}

Note that many algorithms using angles use normalized vectors. Therefore, most of the normalization cost when computing the basis $(\vec{i},\vec{j})$ is not an additional cost of multi-dimensional routing.

\subsection{Experimentation}

The routing protocols GRIC~\cite{PS07} and ROAM~\cite{HJM+09} were experimented in \cite{Vrac1}, using 6 anchors randomly distributed in the network. Messages were delivered with 100\% success, and the average length of the paths followed by messages using anchor coordinates was the same as when using physical Euclidean coordinates.

\section{Algebraic Analysis} %% Formal Analysis
\label{sec_algebra}

In the plane with Euclidean distance, any node has a pair of physical coordinates $X=(x,y)$. We denote by $A_i = (x_i,y_i)$ the physical coordinates of the $i^{th}$ anchor. The anchor coordinates function is a function from $\mathbb{R}^2\rightarrow\mathbb{R}^n$ defined by
$$f:(x,y)\rightarrow \begin{pmatrix} \sqrt{(x-x_1)^2 + (y-y_1)^2} \\ \sqrt{(x-x_2)^2 + (y-y_2)^2} \\ \dots \\ \sqrt{(x-x_n)^2 + (y-y_n)^2} \end{pmatrix}.$$
Since the functions $f_i:(x,y)\rightarrow\sqrt{(x-x_i)^2 + (y-y_i)^2}$ are continuous and $C^\infty$ except in $(x_i,y_i)$, we show that with three or more anchors that are not on the same line, the image $f(\mathbb{R}^2)$ in $\mathbb{R}^n$ is a {\em continuous surface} (claim~\ref{claim_tangent}).
Figure~\ref{fig:3a} represents the image of $f$, when there are three anchors at location $(0,0)$, $(0,1)$ and $(1,0)$.

\begin{figure}[tb]
 \begin{center}
  \includegraphics[width=8cm]{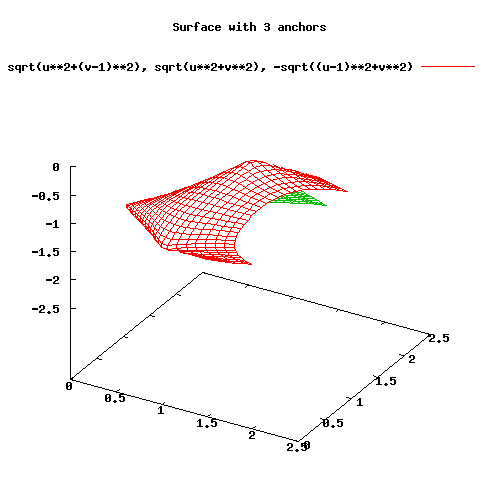}
  \caption{Representation of the distance to three anchors.}
  \label{fig:3a}
 \end{center}
\end{figure}

First, we describe in subsection \ref{sec_tangent} the vector spaces that are tangent to $f(\mathbb{R}^2)$.
Next, we express in subsection \ref{sec_vector} what is the physical direction of messages that use the greedy strategy with virtual coordinates. This physical direction produces a curve that approximates the paths followed by messages.
We then tell in subsection \ref{sec_virtual} what are the convergence conditions on $f(\mathbb{R}^2)$ under which the curve ends at the destination, and prove a bound on the length of this curve.
Finally, we study in subsection \ref{sec_physical} how the placement of anchors affect the convergence conditions and how we can guarantee that they are met.

\subsection{Tangent Space}
\label{sec_tangent}
At any point $f(x,y)$, the surface $f(\mathbb{R}^2)$ has a tangent vector space spanned by the two vectors $\frac{\partial f}{\partial x}(x,y)$ and $\frac{\partial f}{\partial y}(x,y)$. We have

$$\frac{\partial f}{\partial x}(x,y) =
\begin{pmatrix}
  \frac{x-x_1}{\sqrt{(x-x_1)^2+(y-y_1)^2}} \\
  \frac{x-x_2}{\sqrt{(x-x_2)^2+(y-y_2)^2}} \\
  \dots \\
  \frac{x-x_n}{\sqrt{(x-x_n)^2+(y-y_n)^2}}
\end{pmatrix}$$
and
$$\frac{\partial f}{\partial y}(x,y) =
\begin{pmatrix}
  \frac{y-y_1}{\sqrt{(x-x_1)^2+(y-y_1)^2}} \\
  \frac{y-y_2}{\sqrt{(x-x_2)^2+(y-y_2)^2}} \\
  \dots \\
  \frac{y-y_n}{\sqrt{(x-x_n)^2+(y-y_n)^2}}
\end{pmatrix}.$$

\begin{claim}
\label{claim_tangent}
The vector space that is tangent to the surface $f(\mathbb{R}^2)$ in $f(X)$ is two-dimensional if and only if the node $X$ and the anchors $A_1, A_2,\dots, A_n$ are not situated on a single line in the physical space.
\end{claim}

\begin{proof}
The tangent vector space is two-dimensional if and only if $\frac{\partial f}{\partial x}(x,y)$ and $\frac{\partial f}{\partial y}(x,y)$ are not collinear.
Conversely $\frac{\partial f}{\partial x}(x,y)$ and $\frac{\partial f}{\partial y}(x,y)$ are collinear if and only if there is $\alpha\in[0,2\pi[$ such that $\frac{\partial f}{\partial x}(x,y)\cos\alpha + \frac{\partial f}{\partial y}(x,y)\sin\alpha = 0$. By changing the physical coordinates into $u = x\cos\alpha + y\sin\alpha$ and $v = y\cos\alpha - x\sin\alpha$ (we also set $u_i = x_i\cos\alpha + y_i\sin\alpha$ and $v_i = y_i\cos\alpha - x_i\sin\alpha$), we express the tangent vector space with the two vectors

$$\frac{\partial f}{\partial u}(X) =
\begin{pmatrix}
  \frac{u-u_1}{\sqrt{(u-u_1)^2+(v-v_1)^2}} \\
  \frac{u-u_2}{\sqrt{(u-u_2)^2+(v-v_2)^2}} \\
  \dots \\
  \frac{u-u_n}{\sqrt{(u-u_n)^2+(v-v_n)^2}}
\end{pmatrix}$$
and
$$\frac{\partial f}{\partial v}(X) =
\begin{pmatrix}
  \frac{v-v_1}{\sqrt{(u-u_1)^2+(v-v_1)^2}} \\
  \frac{v-v_2}{\sqrt{(u-u_2)^2+(v-v_2)^2}} \\
  \dots \\
  \frac{v-v_n}{\sqrt{(u-u_n)^2+(v-v_n)^2}}
\end{pmatrix}.$$

We have $\frac{\partial f}{\partial u}(X) = 0$ if and only if for all $i\in\{1\dots n\}, u=u_i$.
\end{proof}

When  $\frac{\partial f}{\partial x}(x,y)$ and $\frac{\partial f}{\partial y}(x,y)$ are not collinear, then the Jacobian matrix
$$J_f(X) = J_f(x,y) =
\begin{pmatrix}
  \frac{x-x_1}{\sqrt{(x-x_1)^2+(y-y_1)^2}}
  & \frac{y-y_1}{\sqrt{(x-x_1)^2+(y-y_1)^2}}\\
  \frac{x-x_2}{\sqrt{(x-x_2)^2+(y-y_2)^2}}
  & \frac{y-y_2}{\sqrt{(x-x_2)^2+(y-y_2)^2}} \\
  \dots & \dots \\
  \frac{x-x_n}{\sqrt{(x-x_n)^2+(y-y_n)^2}}
  & \frac{y-y_n}{\sqrt{(x-x_n)^2+(y-y_n)^2}}
\end{pmatrix}.$$
defines a morphism of the physical plane into the vector space tangent to $f(\mathbb{R}^2)$ at $f(x,y)$. Given a node at position $X$ in the physical space and its neighbors at position $X_1, X_2,\dots, X_\delta$, it is not unreasonable to assume that for all $i$, $f(X_i)$ is close to the Taylor expansion $f(X) + J_f(X)(\vec{XX_i})$ in the affine space tangent to $f(\mathbb{R}^2)$ in $f(X)$.

\subsection{Directional Vector}
\label{sec_vector}
In a greedy routing strategy using virtual coordinates, the neighbor $X'$ of choice will be a maximum for some scalar product $\vec{f(X)f(X')}\cdot\vec{f(X)f(D)}$.

\begin{claim}
Given two physical positions $X,D\in\mathbb{R}^2$, the function $s_X:\mathbb{R}^2\rightarrow\mathbb{R}$ such that for any vector $\vec{XX'}\in\mathbb{R}^2$, $s_X(\vec{XX'})$ is the scalar product of $J_f(X)(\vec{XX'})$ by $\vec{f(X)f(D)}$ is a linear form that can be expressed as
$$\vec{XX'}\rightarrow\vec{XX'}\cdot\sum_i{\alpha_i \vec{XA_i}}$$
where $\alpha_i = \frac{d(X,A_i)-d(D,A_i)}{d(X,A_i)}$.
\end{claim}

\begin{proof}
The transformation $\vec{XX'}\rightarrow J_f(X)(\vec{XX'})$ is a linear function. Since the scalar product by $\vec{f(X)f(D)}$ is a linear form, $s_X$ is also a linear form.
We may decompose the vector $\vec{f(X)f(D)}$ into $\sum_i{(d(D,A_i) - d(X,A_i)) \one_i}$ where $\one_i$ is the multi-dimensional vector with 1 as its $i^{th}$ coordinate and zeroes everywhere else.
In this manner, $s_X =\sum_i{s_{X,i}}$ where
$$s_{X,i}(\vec{XX'}) = (d(D,A_i) - d(X,A_i)) J_f(X)(\vec{XX'})\cdot \one_i$$
$$J_f(X)(\vec{XX'})\cdot \one_i = \frac{(x-x_i)(x'-x)+(y-y_i)(y'-y)}{\sqrt{(x-x_i)^2+(y-y_i)^2}}.$$
Thus $s_{X,i}$ can be expressed as
$$\vec{XX'}\rightarrow\vec{XX'}\cdot\frac{d(X,A_i)-d(D,A_i)}{d(X,A_i)} \vec{XA_i}.$$
\end{proof}

Given a node at physical location $X$ and a destination $D\in\mathbb{R}^2$, we call {\em apparent destination} related to $D$ in $X$ the location
$$D' = X + \sum_i\alpha_i \vec{XA_i} =
 X + \sum_i{\frac{d(X,A_i)-d(D,A_i)}{d(X,A_i)} \vec{XA_i}}.$$

\subsection{Virtual consistency}
\label{sec_virtual}
We say that the anchor coordinate system is virtually consistent at distance $r$ for a physical destination $D\in\mathbb{R}^2$, if at every point $X\neq D$ such that $f(X)$ is in a closed metric ball of center $f(D)$ and radius $r$, then $s_X \neq 0$. Note that $s_x = 0$ if and only if the multi-dimensional vector $\vec{f(X)f(D)}$ is orthogonal to the vector space tangent to $f(\mathbb{R}^2)$ in $f(X)$. It is also equivalent to state that the anchor coordinate system is virtually consistent at distance $r$ for a physical destination $D\in\mathbb{R}^2$, if no closed metric ball centered on $f(D)$ and of radius $0<r'\leq r$ is tangent to $f(\mathbb{R}^2)$.

\begin{claim}
If the anchor coordinate system is virtually consistent at distance $r$ for a physical destination $D\in\mathbb{R}^2$, then
there is $\lambda\in\mathbb{R}^+$ such that
for any point $X_0$ with $f(X_0)$ in a closed metric ball of center $f(D)$ and radius $r$ we have a curve $c[0,1]\in\mathbb{R}^2$ that verifies:
\begin{itemize}
\item $c:[0,1]\rightarrow\mathbb{R}^2$ is a derivable function,
\item $c(0) = X_0$ and $c(1) = D$,
\item At any point $t\in[0,1[$, the vector $\frac{\partial c}{\partial t}(t)$ is collinear with the vector $\vec{c(t)D_t'}$ where $D_t'$ is the apparent destination related to $D$ in $c(t)$.
\item $\int_0^1{||\frac{\partial c}{\partial(t)}||dt}\leq \lambda d(X_0,D)$.
\end{itemize}
\end{claim}

\begin{proof}
Let $k$ be the largest positive number such that for any point $X=(x,y)$ with $f(X)$ in a closed metric ball of center $f(D)$ and radius $r$,
%$||s_X|| \geq k$.
the orthogonal projection of $\vec{f(X)f(D)}$ on the vector space defined by the two vectors $\frac{\partial f}{\partial x}(x,y)$ and $\frac{\partial f}{\partial y}(x,y)$ has a norm greater than or equal to $k d(X,D)$.
Since $f$ is a continuous function, the set of physical positions $X$ such that $d(f(X),f(D))\leq r$ is compact subset of $\mathbb{R}^2$. Therefore, if $k$ was equal to zero, then there would be a point $X\neq D$ in the ball such that
%$s_X = 0$,
$\vec{XD}$ is orthogonal to the surface $f(\mathbb{R}^2)$,
which we excluded in our assumptions.

Let $c:[0,1]\rightarrow\mathbb{R}^2$ be the function defined by $c(0) = X_0$ and such that
%$\frac{\partial c}{\partial t}(t) = \vec{XD'}$, where $X = c(t)$ and $D'$ is the apparent destination related to $D$ in $X$.
$\frac{\partial (f \circ c)}{\partial t}(t)$ is the orthogonal projection of $\frac{k^{-2}d(f(X_0),f(D))}{d((f\circ c)(t),f(D))}\vec{(f\circ c)(t)f(D)}$ on the vector space defined by the two vectors $\frac{\partial f}{\partial x}(c(t))$ and $\frac{\partial f}{\partial y}(c(t))$.
%We have
%$$\frac{\partial d(f(c(t)),f(D))}{\partial t}(t) = $||J_f(c(t))(
Since
$$\frac{\partial (f \circ c)}{\partial t}(t)\cdot\frac{\vec{(f\circ c)(t)f(D)}}{||\vec{(f\circ c)(t)f(D)}||} \geq k||\frac{\partial (f \circ c)}{\partial t}(t)||$$
we can see that
$$\frac{\partial d((f\circ c)t),f(D))}{\partial t}(t)\geq d(f(X_0),f(D))$$
which implies that $c(1)= D$.
The norm of $\frac{\partial c}{\partial t}(t)$ is smaller than or equal to $||(J_f(c(t)))^{-1}||d(f(X_0),f(D))$, which means that
$$\int_0^1{||\frac{\partial c}{\partial(t)}||dt}\leq \max_{t\in[0,1]}||(J_f(c(t)))^{-1}||d(f(X_0,f(D)))$$
$$\int_0^1{||\frac{\partial c}{\partial(t)}||dt}\leq \sqrt{n}\max_{t\in[0,1]}||(J_f(c(t)))^{-1}||d(X_0,D).$$
\end{proof}

\subsection{Physical consistency}
\label{sec_physical}
We say that the anchor coordinate system is physically consistent at position $X$ for the destination $D$ if $\vec{XD'}\cdot\vec{XD} > 0$, where $D'$ is the apparent destination related to $D'$ in $X$. Observe that if the anchor coordinate system is physically consistent for the destination $D$ in a ball $\cal B$ around $D$, then it is virtually consistent at distance $r$ for the physical destination $D$, where $r$ is the radius of the biggest multi-dimensional ball $\Omega$ such that $\Omega\cap f(\mathbb{R}^2)\subset f(\cal B)$.

To study the physical consistency of the system at position $X$ for the destination $D$, we split the physical plane in four parts $P_1, P_2, P_3, P_4$ with
$P_1 = \{X' | \vec{XX'}\cdot\vec{XD}\leq 0 \}$,
$P_2 = \{X' | \vec{XX'}\cdot\vec{XD}> 0$ and $d(X,X') < d(X,D) \}$,
$P_3 = \{X' | \vec{DX'}\cdot\vec{DX}> 0$ and $d(X,X') \geq d(X,D) \}$,
$P_4 = \{X' | \vec{DX'}\cdot\vec{DX}\leq 0 \}$.
Since the apparent destination $D'$ is defined by
$$D' = X + \sum_i{\frac{d(X,A_i)-d(D,A_i)}{d(X,A_i)} \vec{XA_i}}$$
we see as illustrated in Figure~\ref{fig_consistency} that only the anchors in $P_2$ give a negative contribution to $\vec{XD'}\cdot\vec{XD}$.
\begin{figure}[tb]
 \begin{center}
  \includegraphics[width=8cm]{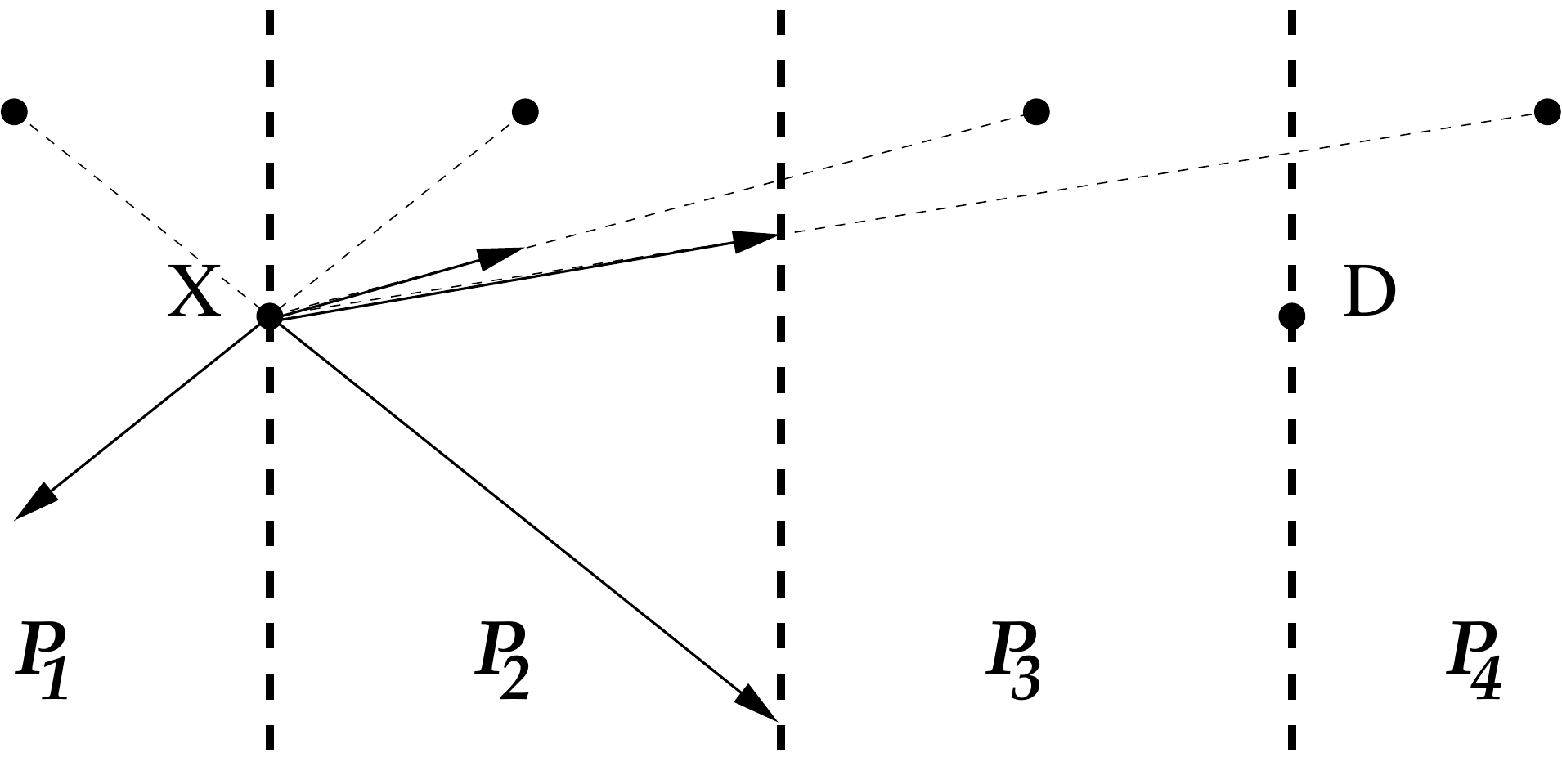}
  \caption{Contribution of anchors in $P_1, P_2, P_3, P_4$.}
  \label{fig_consistency}
 \end{center}
\end{figure}

If anchors are randomly distributed in the network, the negative contribution will most probably be small enough for the system to be consistent, unless $P_1$ and $P_4$ are almost void of nodes, which happens when $X$ and $D$ are located on opposite borders of the network (so that all the anchors are between them). This situation never occurred in the experiments carried out in~\cite{Vrac1} (but the end nodes were not chosen on the border of the network).
Nevertheless, physical inconsistency may be avoided by selecting anchors when the destination $D$ of a message originating from $X_0$ is far away:
\begin{enumerate}
\item by default, use all the anchors.
\item compute $l_A = \max_{i\in\{1,..,n\}} \max(d(D,A_i),d(X_0,A_i)) = \max(||f(D)||_\infty,||f(X_0)||_\infty)$. $l_A$ gives an idea of the diameter of the network.
\item for each node $X$ along the path of the message
  \begin{enumerate}
  \item compute $l_X = \max_{i\in\{1,..,n\}} |d(X,D) - d(X,A_i)| = ||f(D)-f(X)||_\infty$. $l_X$ is smaller than $d(X,D)$.
  \item if using all the anchors and if  $l_X > \frac{2l_A}{3}$ then use only the anchors $A_i$ such that $d(D,A_i) < \frac{l_A}{3}$.
  \item if using a subset of anchors and if $l_X < \frac{l_A}{2}$ then use all the anchors.
  \end{enumerate}
\end{enumerate}
In this way, physical inconsistency can be completely avoided in the network, at the cost of using a different coordinate system when $d(X,D)$ is comparable to the diameter of the network.

\section{Conclusion}
Geographic routing is an essential component in connecting sensor networks. Foregoing the previously necessary localization phase where physical Cartesian coordinates are produced is an important step into making networks more robust and totally independent from external hardware. Sensor network applications that use localization information exclusively inside the network may transparently use virtual coordinates, whereas sophisticated physical localization may still be performed at some external base station from the virtual coordinates whenever localization must be used externally. In this way, directly using raw distance information without any costly or sophisticated localization calculus is a simple, viable, and efficient way to perform geographic routing.

\bibliographystyle{abbrv}
\bibliography{virtualRoutingSac}
\end{document}